\documentclass[submission,copyright,creativecommons]{eptcs}
\usepackage[utf8]{inputenc}
\usepackage[english]{babel}
\usepackage{amsmath,amsfonts,amssymb,amsthm}
\usepackage{newunicodechar}
\usepackage{tikz}\usetikzlibrary{patterns}
\usepackage{tikz-cd}
\usepackage{mathpartir}
\usepackage[only,llbracket,rrbracket]{stmaryrd}
\usepackage{macros}
\usepackage{hyperref}
\usepackage[capitalize]{cleveref}

\newtheorem{theorem}{Theorem}
\newtheorem{proposition}[theorem]{Proposition}
\newtheorem{lemma}[theorem]{Lemma}
\newtheorem{corollary}[theorem]{Corollary}
\theoremstyle{definition}
\newtheorem{definition}[theorem]{Definition}
\theoremstyle{remark}
\newtheorem{remark}[theorem]{Remark}
\newtheorem{example}[theorem]{Example}

\crefname{def}{Definition}{Definitions}
\crefname{thm}{Theorem}{Theorems}
\crefname{ex}{Example}{Examples}
\crefname{rem}{Remark}{Remarks}
\crefname{prop}{Proposition}{Propositions}
\crefname{lem}{Lemma}{Lemmas}

\newunicodechar{∥}{\ensuremath\|}

\hypersetup{hidelinks}
\title{Syntactic Regions for Concurrent Programs}
\author{
  Samuel Mimram
  \institute{École polytechnique}
  \and
  Aly-Bora Ulusoy
  \institute{École polytechnique}
}

\def\authorrunning{S. Mimram, A.-B. Ulusoy}
\def\titlerunning{Syntactic Regions for Concurrent Programs}
\providecommand{\event}{MFPS XXXVII}

\begin{document}
\maketitle

\begin{abstract}
  In order to gain a better understanding of the state space of programs, with the
aim of making their verification more tractable, models based on directed
topological spaces have been introduced, allowing to take in account equivalence
between execution traces, as well as translate features of the execution (such
as the presence of deadlocks) into geometrical situations. In this context, many
algorithms were introduced, based on a description of the geometrical models as
regions consisting of unions of rectangles. We explain here that these
constructions can actually be performed directly on the syntax of programs, thus
resulting in representations which are more natural and easier to implement. In
order to do so, we start from the observation that positions in a program can be
described as partial explorations of the program. The operational semantics
induces a partial order on positions, and regions can be defined as formal
unions of intervals in the resulting poset. We then study the structure of such
regions and show that, under reasonable conditions, they form a boolean algebra
and admit a representation in normal form (which corresponds to covering a space
by maximal intervals), thus supporting the constructions needed for the purpose
of studying programs. All the operations involved here are given explicit
algorithmic descriptions.


\end{abstract}

\section{Introduction}
The verification of concurrent programs is notoriously complicated because of
the combinatorics involved when performing a naive transposition of the
techniques available for sequential programs. Namely, for a program consisting of
multiple processes running in parallel, we have to make sure that no problem can
arise for whichever respective scheduling of the processes, and the number of
resulting execution traces to check is in general exponential in the size of the
original program: this is sometimes called the ``state-space explosion
problem''. Moreover, some errors such as the presence of deadlocks are specific
to concurrent programs and are not addressed by traditional techniques. In order
to tackle these issues, starting in the 90s, a series of theoretical and
practical tools have been introduced based on the \emph{geometric semantics} of
concurrent programs~\cite{goubault2003some,goubault2005practical,fajstrup2006algebraic,datc,fajstrup2012trace}, which assigns to each such
program a topological space, such that the possible states of the program can be
interpreted as points in this space and an execution as a path which is
\emph{directed}, \ie intuitively respects the direction of time. An important
observation is that, in those semantics, two paths which are dihomotopic (can be
continuously deformed one to the other while respecting the direction of time)
correspond to executions which are equivalent from an operational point of view,
and we can thus hope to reduce the state space of the program by considering
paths up to dihomotopy. This approach is obviously inspired by the one taken in
the algebraic study of spaces (which considers algebraic constructions which
are invariant up to homotopy, such as the fundamental and higher homotopy
groups), although the situation for programs is made more difficult because one
has to take the direction of time in account.

While we believe that this line of research is conceptually enlightening and is
very fruitful, the fact that one has to resort to topological spaces in order to
study discrete structures such as the ones offered by programs is quite puzzling
and it is natural to wonder if the topology is really necessary here. It turns
out that it is not, and one of the main goals of this paper is to translate into
purely combinatorial terms an important algorithmic construction in geometric
semantics: the one of \emph{cubical regions}. We refer the reader to
\cite[Chapter~5]{datc} for a detailed presentation but, roughly, it consists in
describing the geometric semantics by covering it with cubes: in such a cube,
any two paths are homotopic, and it is thus enough to check only one path in
order to perform verification.
%
%
The starting point of this paper consists in considering that an execution of a
program consists in a partial exploration of a ``prefix'' of this program, that
we call here a \emph{position} (on the logical side, similar ideas have been
advocated by Girard when defining ludics~\cite{girard2001locus}). A ``portion''
of the program can then be identified by an interval of positions, which plays
the role of the cubes above: it denotes all the positions where we have explored
more than the first and less than the second. We can finally reach a convenient
description of regions in the program by taking finite unions of such intervals,
that we call here \emph{syntactic regions}. We formally define those here and
study their properties. Given a set of positions, there are in general multiple
regions which describe this set of positions: we show that, under mild assumptions, there is
always a canonical region describing a set of positions, the \emph{normal
  region}, which is in some sense the most economical way of describing
it. Moreover, we show that such regions form a boolean algebra
(\cref{thm:nf-ba}) and provide explicit ways of computing corresponding
canonical operations (union, intersection, complement) on syntactic regions,
which is useful in practice. Typically, we can compute a representation of the
state space of a concurrent program by first computing the region which is
forbidden because of the use of mutual exclusion primitives, and then compute
the state space as the complement of this region, on which we can use the
traditional techniques mentioned above.

In addition to providing us with a better understanding of the ``geometry of
programs'', the aim of the techniques developed is to reduce the size
of the state space in order to perform program verification, in a similar vein as partial order reduction
(por)~\cite{Godefroid95:por}. A detailed comparison between the two has been
performed both from a practical point of view~\cite{ts} and a theoretical
one~\cite{goubault:geom-por}: while the two techniques perform similarly on basic
examples, geometric techniques sometimes outperform por. Our approach aims at
being fully automated and is thus less powerful than techniques based on logic,
where some user input or annotations are required (such as the Owicki-Gries
method~\cite{owicki1976axiomatic}, rely/guarantee
rules~\cite{coleman2007structural}, concurrent separation
logic~\cite{o2007resources}, etc.). The effectiveness of our approach,
which is simpler and easier to implement than previous ones, is
illustrated with the online prototype~\cite{git}, where concrete examples of
programs on which it applies are provided (typically, producer-consumer
algorithms).



We begin by introducing the programming language considered here as well as an original
notion of position on its programs in \cref{sec:prog}, we then define the state space
of concurrent programs in \cref{sec:mutex} and explain how it can be described
using geometrical regions in \cref{sec:geom-sem}. We generalize the notion of
region to arbitrary posets in \cref{sec:region}, characterize when regions in
normal forms have a structure of boolean algebra in \cref{sec:boolean-algebra}
and finally provide explicit constructions for regions coming from programs in
\cref{sec:syntactic}.

\section{Concurrent programs and their positions}
\label{sec:prog}
In order to abstract away from practical details, we introduce here a simple
concurrent imperative programming language. We will only be interested in the
control-flow structure and thus the operations we use are not relevant for our
matters: we simply suppose fixed a set $\actions=\set{A,B,\ldots}$ of
\emph{actions}, which can be thought of as the effectful operations of our
language, such as modifying a variable or printing a result. The present work
would extend to more realistic languages without any difficulty directly related
to the main points raised here.

\begin{definition}
  \label[def]{def:program}
  The collection of \emph{programs}~$P$ is generated by the following grammar:
  \[
    P,Q
    \gramdef
    A
    \gramor
    \pseq PQ
    \gramor
    \ploop P
    \gramor
    \por PQ
    \gramor
    \ppar PQ
  \]
\end{definition}

\noindent
A program is thus either an action~$A$, or a sequential composition $\pseq PQ$
of two programs $P$ and~$Q$, or a conditional branching $\por PQ$ which will
execute either~$P$ or~$Q$, or a conditional loop $\ploop P$ which will
execute~$P$ a given number of times, or the execution $\ppar PQ$ of two
subprograms~$P$ and~$Q$ in parallel. As explained above, we do not take
variables into account, so that branching and looping is non-deterministic, but we
could handle proper conditional branching and while loops.

A position in a program describes where we are during an execution of it and
thus encodes the ``prefix'' of the program which has already been
executed. Formally, we begin by the following definition:

\begin{definition}
  The \emph{pre-positions}~$p$ are generated by the following grammar, with
  $n\in\N$:
  \[
    p,q
    \gramdef
    \bot
    \gramor
    \ptop
    \gramor
    \pseq pq
    \gramor
    \ploop[n]p
    \gramor
    \por pq
    \gramor
    \ppar pq
  \]
\end{definition}

\noindent
Those can be read as: we have not started (\resp we have finished) the execution
($\pbot$, \resp $\ptop$),
we are executing a sequence ($\pseq pq$), we are in the $n$-th iteration of a
loop ($\ploop[n]p$), we are executing a branch of a conditional branching
($\por pq$) and we are executing two programs in parallel~($\ppar pq$). Note
that the syntax of positions is essentially the same as the one of programs,
except that actions have been replaced by $\bot$ and $\top$ (and loops are
``unfolded'' in the sense that we keep track of the loop number).


Next, we single out the positions which are valid for a program. For instance,
we want that, in a program of the form~$\pseq PQ$, we can begin executing~$Q$
only after $P$ has been fully executed: this means that a position of the form
$\pseq pq$ with $q\neq\pbot$ is valid only when $p$ is~$\ptop$. Similarly, in a
conditional branching $\por PQ$, we cannot execute both subprograms: a
position~$\por pq$ is valid only when either~$p$ or~$q$ is $\pbot$.

\begin{definition}
  \label[def]{def:position}
  We write $\vpos Pp$ to indicate that a pre-position $p$ is a (valid)
  \emph{position} of a program~$P$, this predicate being defined inductively by
  the following rules:
  \[
    \begin{array}{c@{\qquad\quad}c@{\qquad\quad}c@{\qquad\quad}c}
      \inferrule{ }{\vpos P\pbot}
      &
      \inferrule{\vpos Pp}{\vpos{\pseq PQ}{\pseq p\pbot}}
      &
      \inferrule{\vpos Pp}{\vpos{\por PQ}{\por p\pbot}}
      &
      \inferrule{\vpos Pp}{\vpos{\ploop P}{\ploop[n]p}}
      \\[3ex]
      \inferrule{ }{\vpos P\ptop}
      &
      \inferrule{\vpos Qq}{\vpos{\pseq PQ}{\pseq\ptop q}}
      &
      \inferrule{\vpos Qq}{\vpos{\por PQ}{\por \pbot q}}
      &
      \inferrule{\vpos Pp\\\vpos Qq}{\vpos{\ppar PQ}{\ppar pq}}
    \end{array}
  \]
  We write $\positions(P)$ for the set of positions of a program~$P$.
\end{definition}

\noindent
We can assimilate a position in a program to a possible state during its
execution. With this point of view in mind, we can formalize the operational
semantics of our programming language as a relation $\vred Pp{p'}$, which
can be read as the fact that, in the position $p$, the program $P$ can reduce
in one step and reach the position~$p'$.

\begin{definition}
  \label[def]{def:reduction}
  The \emph{reduction} relation is defined inductively by
  {\small
  \[
    \begin{array}{c@{\qquad}c@{\qquad}c@{\qquad}c}
      \inferrule{ }{\vred A\pbot\ptop}
      \\[3ex]
      \inferrule{ }{\vred{\pseq PQ}\pbot{\pseq\pbot\pbot}}
      &
      \inferrule{ }{\vred{\por PQ}\pbot{\por\pbot\pbot}}
      &
      \inferrule{ }{\vred{\ploop P}\pbot{\ploop[0]\pbot}}
      &
      \inferrule{ }{\vred{\ppar PQ}{\pbot}{\ppar\pbot\pbot}}
      \\[3ex]
      \inferrule{\vred Pp{p'}}{\vred{\pseq PQ}{\pseq p\pbot}{\pseq{p'}\pbot}}
      &
      \inferrule{\vred Pp{p'}}{\vred{\por PQ}{\por p\pbot}{\por{p'}\pbot}}
      &
      \inferrule{\vred Pp{p'}}{\vred{\ploop P}{\ploop[n]p}{\ploop[n]{p'}}}
      &
      \inferrule{\vred P{p}{p'}\\\vpos Qq}{\vred{\ppar PQ}{\ppar pq}{\ppar{p'}q}}
      \\[3ex]
      \inferrule{\vred Qq{q'}}{\vred{\pseq PQ}{\pseq\ptop q}{\pseq\ptop{q'}}}
      &
      \inferrule{\vred Qq{q'}}{\vred{\por PQ}{\por \pbot q}{\por\pbot{q'}}}
      &
      \inferrule{ }{\vred{\ploop P}{\ploop[n]\ptop}{\ploop[n+1]\pbot}}
      &
      \inferrule{\vpos Pp\\\vred Q{q}{q'}}{\vred{\ppar PQ}{\ppar pq}{\ppar p{q'}}}
      \\[3ex]
      \inferrule{ }{\vred{\pseq PQ}{\pseq\ptop\ptop}\ptop}
      &
      \inferrule{ }{\vred{\por PQ}{\por\ptop\pbot}\ptop}
      &
      \inferrule{ }{\vred{\ploop P}{\ploop[n]\ptop}{\ptop}}
      &
      \inferrule{ }{\vred{\ppar PQ}{\ppar\ptop\ptop}\ptop}
      \\[3ex]
      &
      \inferrule{ }{\vred{\por PQ}{\por\pbot\ptop}\ptop}
      &
      \inferrule{ }{\vred{\ploop P}\pbot\ptop}
    \end{array}
  \]}
\end{definition}

\noindent
The above operational semantics is ``very fine-grained'' in the sense that it
features transitions which are not usually observable, such as
$\vred{\pseq PQ}\pbot{\pseq\pbot\pbot}$ which corresponds to passing from a
state where we have not yet started executing the program to a state where we
have started executing a sequence, but not yet its components. The usual
``real'' actions correspond to executions of the upper left rule
$\vred{A}\pbot\ptop$ which can be interpreted as executing an action~$A$.

\begin{lemma}
  \label[lem]{lem:red}
  If $\vred Pp{p'}$ holds then both $\vpos Pp$ and $\vpos P{p'}$ hold.
\end{lemma}

Suppose fixed a program~$P$. The \emph{state space}~$\sspace P$ of this program
is the graph whose vertices are the positions of~$P$ (\cref{def:position}) and
edges are the reductions (\cref{def:reduction}). An \emph{execution path} of~$P$
is a morphism of the free category~$\scat P$ over the graph~$\sspace P$. We
write $\vpath P\pi p{p'}$ (or simply $\pi:p\pathto q$) to indicate that $\pi$ is
an execution of~$P$ from $p$ to~$q$, we write $\pi\cdot\pi'$ for the composition
(also called \emph{concatenation}) of two paths $\pi:p\pathto p'$ and
$\pi':p'\pathto p''$, and we write $\varepsilon:p\pathto p$ for the identity
execution (also called \emph{empty path}).
An execution is \emph{elementary} when it consists of one reduction step. A
\emph{global execution} is an execution from~$\bot$ and a \emph{total execution}
is a global execution to $\top$.
We say that an execution~$\pi'$ is a \emph{prefix} of an execution $\pi$ when
there exists an execution $\pi''$ such that $\pi=\pi'\cdot\pi''$.
A position is \emph{reachable} when there exists a global path
$\pi:\pbot\pathto p$ with this position as target. The following lemma, together
with \cref{lem:red}, formalizes the fact the notion of validity of
\cref{def:position} captures exactly the expected positions.

\begin{lemma}
  \label[lem]{lem:reachable}
  Every position $p$ of~$P$ is reachable.
\end{lemma}

As customary, we also write $\vreds Pp{p'}$ (or sometimes simply $p\pathto p'$)
when there exists an execution $\vpath P\pi p{p'}$: the resulting relation
$\pathto$ on positions of~$P$ is the reflexive and transitive closure of the
reduction relation~$\to$. This relation is induced by a partial order relation,
as we now explain.

\begin{definition}
  \label[def]{def:pos-order}
  We write $\leq$ for the smallest reflexive relation on the positions of~$P$
  such that
  \[
    \begin{array}{c@{\qquad\quad}c@{\qquad\quad}c@{\qquad\quad}c}
      \inferrule{ }{\pbot\leq p}
      &
      \inferrule{p\leq p'\\q\leq q'}{\pseq pq\leq\pseq{p'}{q'}}
      &
      \inferrule{p\leq p'\\q\leq q'}{\ppar pq\leq\ppar{p'}{q'}}
      &
      \inferrule{p\leq p'}{\ploop[n]p\leq\ploop[n]{p'}}
      \\
      \inferrule{ }{p\leq\ptop}
      &
      \inferrule{p\leq p'\\q\leq q'}{\por pq\leq\por{p'}{q'}}
      &
      &
      \inferrule{ }{\ploop[m]p\leq\ploop[n]{p'}}
    \end{array}
  \]
  for $m<n$.
\end{definition}

\begin{proposition}
\label[prop]{prop:sem-equi}
  Given two positions $p$ and $p'$ of~$P$, we have $p\leq p'$ if and only
  if $p\pathto p'$.
\end{proposition}

\begin{proposition}
  \label[prop]{prop:pos-poset}
  The relation $\leq$ is a partial order on the set $\positions(P)$ of positions
  of~$P$.
\end{proposition}

\begin{proposition}
  \label[prop]{prop:lat-op}
  The partial order $\leq$ on $\positions(P)$ is a bounded lattice, with
  $\pbot$ and $\ptop$ as smallest and largest elements, with supremum being
  determined by
  \begin{align*}
    (\pseq pq)\vee(\pseq{p'}{q'})&=\pseq{(p\vee p')}{(q\vee q')}
    &
    \ploop[n]{p}\vee\ploop[n]{p'}&=\ploop[n]{(p\vee p')}
    &
    (\por p\pbot)\vee(\por{p'}\pbot)&=\por{(p\vee p')}{\pbot}
    \\
    (\ppar pq)\vee(\ppar{p'}{q'})&=\ppar{(p\vee p')}{(q\vee q')}
    &
    \ploop[m]p\vee\ploop[n]{p'}&=\ploop[n]{p'}
    &
    (\por \pbot q)\vee(\por\pbot{q'})&=\por{\pbot}{(q\vee q')}
    \\
    &&&&
    (\por p\pbot)\vee(\por\pbot q)&=\ptop
  \end{align*}
  for $m<n$, and where we suppose $(p,q)\neq(\pbot,\pbot)$ in the lower right
  rule. The infimum admits a similar description.
\end{proposition}

We recall that a poset $(X,\leq)$ is a \emph{well-order} when, for every infinite
sequence $(x_i)_{i\in\N}$ of elements, there are indices $i<j$ such that
$x_i\leq x_j$. This is equivalent to requiring that the poset is both well-founded
(every infinite decreasing sequence of elements is eventually stationary) and such that
every antichain (set of pairwise incomparable elements) is finite. The following
will be useful:

\begin{proposition}
  \label[prop]{prop:pos-wo}
  The poset~$\positions(P)$ is a well-order.
\end{proposition}
\begin{proof}
  The proof proceeds by induction on~$P$. The inductive cases are easily deduced
  from the fact that well-orders are closed under products, coproducts and
  contain $(\N,\leq)$.
\end{proof}

\section{Concurrent programs with mutexes and their state space}
\label{sec:mutex}
In practice, in order to avoid unspecified behaviors due to concurrency,
programs use primitive operations such as mutexes~\cite{dijkstra1965solution}
with the aim of preventing that two subprograms access simultaneously to a
shared resource such as memory (typically, the result of two concurrent writings
at a same memory location is unspecified in many languages). In order to account
for this, we suppose fixed a set $\mutexes$ of \emph{mutexes} such that, for
every $a\in\mutexes$, there are two associated actions $\P a$ and $\V a$ in
$\actions$, respectively corresponding to \emph{taking} and \emph{releasing} the mutex $a$,
as defined in~\cite{dijkstra1965solution}. A
mutex will be such that it can be can be taken by at most one subprogram: if
another subprogram tries to take an already taken resource, it will be
``frozen'' until the mutex is available again, \ie the first subprogram releases
the mutex. This description is formalized below on the operational semantics,
see also~\cite[Chapter~3]{datc}.

The consumption of mutexes by an execution is kept track of by considering
functions in~$\Z^\mutexes$ which to every mutex associate the number of times
it has been taken or released (where releasing is considered as the opposite of
taking). Given two functions $\mu,\mu'\in\Z^\mutexes$, we write $\mu+\mu'$ for
their pointwise sum, \ie $(\mu+\mu')(a)=\mu(a)+\mu'(a)$ for $a\in\mutexes$, and
similarly for the opposite $-\mu$ of a function~$\mu$. Given $a\in\mutexes$, we
write $\delta_a\in\Z^\mutexes$ for the function such that $\delta_a(a)=1$ and
$\delta_b(a)=0$ for $b\neq a$, we also write $\ul 0$ for the constant function
equal to $0$.

\begin{definition}
  Given an execution $\vpath P\pi p{p'}$, we write
  $
    \intp{\vpath P\pi p{p'}}:\mutexes\to\Z
  $
  (or sometimes simply $\intp\pi$) for its \emph{consumption} of mutexes. This
  function is defined for elementary executions by
  \begin{itemize}
  \item $\intp{\vpathred{\P a}\pi\pbot\ptop}=\delta_a$, for $a\in\mutexes$,
  \item $\intp{\vpathred{\V a}\pi\pbot\ptop}=-\delta_a$, for $a\in\mutexes$,
  \item for each reduction $\pi$, different from the two above, deduced with a
    rule of \cref{def:reduction} with no premise we have $\intp{\pi}=\ul 0$,
  \item for each reduction $\pi$ deduced with a rule of \cref{def:reduction}
    with one reduction $\pi'$ as premise, we have $\intp{\pi}=\intp{\pi'}$,
  \end{itemize}
  and extended as an action of the category of executions on the monoid of
  consumptions, \ie $\intp{\varepsilon}=\ul 0$ and
  $\intp{\pi\cdot\pi'}=\intp{\pi'}+\intp{\pi}$.
\end{definition}

\begin{example}
  Writing $\pi$ for any total execution of the program
  $\pseq{\P a}{(\ppar{\P a}{\V b})}$, with $a\neq b$, we have
  $\intp{\pi}(a)=2$, $\intp{\pi}(b)=-1$ and
  $\intp{\pi}(c)=0$ for $c\neq a$ and $c\neq b$.
\end{example}

\noindent
A global execution $\pi$ is \emph{valid} when for every prefix $\pi'$ of~$\pi$,
and any mutex $a\in\mutexes$, we have $0\leq\intp{\pi'}(a)\leq 1$. Such an
execution is namely compatible with the semantics of mutexes in the sense that
\begin{itemize}
\item no mutex is taken twice (without having been released in between),
\item no mutex is released without having been taken first.
\end{itemize}

A program is \emph{conservative} when any two executions $\pi:p\pathto p'$ and
$\pi':p\pathto p'$ with the same source and the same target have the same
resource consumption: $\intp\pi=\intp{\pi'}$. In the following, unless otherwise
stated, we assume that all the programs we consider are conservative: it is
often satisfied in practice and simplifies computations because we can consider
validity of positions instead of executions. This assumption is classical in
geometric semantics~\cite{datc} and is based on the observation that the use of
mutexes is quite costly (because they restrict parallelism) and therefore
usually restricted to small and "local" portions of the code. This is the reason
why it is usually satisfied, at least up to simple rewriting of the code (more
complex code can also be handled by ``duplicating'' some of the positions so
that the program becomes conservative).
This condition can easily be tested as follows.

\begin{definition}
  The \emph{consumption} of a program~$P$ is the partial function
  $\Delta(P):\mutexes\to\Z$ defined by induction on~$P$ by
  \begin{align*}
    \Delta(\P a)&=\delta_a
    &
    \Delta(\V a)&=-\delta_a
    &
    \Delta(A)&=\ul 0
    &
    \Delta(\pseq PQ)&=\Delta(\ppar PQ)=\Delta(P)+\Delta(Q)
  \end{align*}
  and
  \begin{align*}
    \Delta(\por PQ)&=\Delta(P)\quad\text{if $\Delta(P)=\Delta(Q)$}
    &
    \Delta(\ploop P)&=\ul 0\quad\text{if $\Delta(P)=\ul 0$}
  \end{align*}
\end{definition}

\noindent
Note that, because of the side conditions in the two last cases, $\Delta(P)$ is
not always well defined, \eg for $\ploop{\P a}$ or $\por{\P a}{\P b}$.

\begin{proposition}
  \label[prop]{prop:conservative-delta}
  A program~$P$ is conservative if and only if $\Delta(P)$ is well-defined and,
  in this case, we have $\Delta(P)=\intp{\pi}$ for any total execution $\pi$ of~$P$.
\end{proposition}

\noindent
The above proposition gives a simple criterion to check for conservativity, by
induction on programs. It is applicable even when the set~$\mutexes$ of mutexes is
infinite because it is not difficult to show that, for any program~$P$,
$\Delta(P)$ is almost everywhere null when defined, \ie the set
$\setof{a\in\mutexes}{\Delta(P)(a)\neq 0}$ is finite, so that the computations
on consumption can easily be implemented in practice. For conservative programs,
it makes sense to consider the resource consumption at a position (as opposed to
along an execution):

\begin{definition}
  Given a conservative program, we define the
  \emph{consumption}~$\intp{p}:\mutexes\to\Z$ of a position~$p$ as
  $\intp{p}=\intp{\pi}$ for some global path $\pi:\pbot\to p$.
\end{definition}

\noindent
This definition makes sense because there is at least one such path $\pi$ by
\cref{lem:reachable} and it does not depend on the choice of the path by the
assumption of being conservative. This enables us to characterize valid
executions as follows. We say that an execution~$\pi:p\pathto p'$ \emph{visits}
a position~$q$ when $q$ is the target of some prefix~$\pi':p\pathto q$
of~$\pi$. We say that a position~$p$ is \emph{valid} if it satisfies
$0\leq\intp{p}(a)\leq 1$ for every mutex $a\in\mutexes$.

\begin{proposition}
  \label[prop]{prop:exec-pos-validity}
  A global execution~$\pi$ is valid if and only if every position~$p$ it visits
  is valid.
\end{proposition}

\noindent
This motivates defining the \emph{pruned state space}~$\psspace P$ of~$P$ as the
subgraph obtained from~$\sspace P$ by removing every edge between invalid
vertices, and all vertices whose incident edges have all been removed. Writing
$\pscat P$ for the free category it generates, which is a subcategory
of~$\scat P$, we have

\begin{proposition}
  \label[prop]{prop:free-cat}
  The morphisms of~$\pscat P$ are are precisely the valid executions of~$P$.
\end{proposition}


\section{Geometric semantics}
\label{sec:geom-sem}
We now briefly recall (or rather illustrate) the geometric semantics of
concurrent programs, and relate it to the previous point of view on programs. We
refer to the textbook~\cite{datc} as well as~\cite{haucourt:geom-cons} for
further reference on the subject.

In order to take a concrete example, we suppose here only that our actions
allow for traditional manipulations of variables, and consider the program
\[
  \code{($\P a$ ; $\P b$ ; x := y+1 ; $\V b$ ; $\V a$) ∥ ($\P b$ ; $\P a$ ; y := x+2 ; $\V a$ ; $\V b$)}
\]
consisting of two processes executed in parallel. Here, we should think of $a$
and $b$ as protecting the variables \code{x} and \code{y} respectively: the
program applies the discipline that whenever a variable is used, the
corresponding mutex is taken beforehand and released afterward, in order to avoid
any concurrent access to it. The idea behind geometric semantics is that to such
a program we should associate the topological space on the left of
\cref{fig:gsem},
\begin{figure}[t]
  \centering
\[
  \begin{tikzpicture}[baseline=(b.base),scale=.5]
    \coordinate (b) at (0,0);
    \draw (0,0) -- (6,0) -- (6,6) -- (0,6) -- cycle;
    \foreach \x in {1,...,5} \draw (\x,.1) -- (\x,-.1);
    \foreach \y in {1,...,5} \draw (.1,\y) -- (-.1,\y);
    \draw (1,0) node[below] {$\P a$};
    \draw (2,0) node[below] {$\P b$};
    \draw (3,0) node[below] {\rotatebox{-90}{\code{x:=y+1}}};
    \draw (4,0) node[below] {$\V b$};
    \draw (5,0) node[below] {$\V a$};
    \draw (0,1) node[left] {$\P b$};
    \draw (0,2) node[left] {$\P a$};
    \draw (0,3) node[left] {\code{y:=x+2}};
    \draw (0,4) node[left] {$\V a$};
    \draw (0,5) node[left] {$\V b$};
    \filldraw[gray] (1,2) -- (5,2) -- (5,4) -- (1,4) -- cycle;
    \filldraw[gray] (2,1) -- (4,1) -- (4,5) -- (2,5) -- cycle;
    \draw[dotted,rounded corners] (0,0) -- (.5,5) -- (6,6);
    \filldraw (2,2) circle (.1);
  \end{tikzpicture}
  \qquad\qquad
  \begin{tikzpicture}[baseline=(b.base),scale=.5]
    \coordinate (b) at (0,0);
    \draw (0,0) -- (6,0) -- (6,6) -- (0,6) -- cycle;
    \foreach \x in {1,...,5} \draw (\x,.1) -- (\x,-.1);
    \foreach \y in {1,...,5} \draw (.1,\y) -- (-.1,\y);
    \filldraw[pattern=north east lines] (1,2) rectangle (5,4);
    \filldraw[pattern=north west lines] (2,1) rectangle (4,5);
  \end{tikzpicture}
  \qquad\qquad
  \begin{tikzpicture}[baseline=(b.base),scale=.5]
    \coordinate (b) at (0,0);
    \draw (0,0) -- (6,0) -- (6,6) -- (0,6) -- cycle;
    \foreach \x in {1,...,5} \draw (\x,.1) -- (\x,-.1);
    \foreach \y in {1,...,5} \draw (.1,\y) -- (-.1,\y);
    \filldraw[pattern=north east lines] (0,0) rectangle (2,2);
    \filldraw[pattern=north east lines] (0,4) rectangle (2,6);
    \filldraw[pattern=north east lines] (4,4) rectangle (6,6);
    \filldraw[pattern=north east lines] (4,0) rectangle (6,2);
    \filldraw[pattern=horizontal lines] (0,0) rectangle (1,6);
    \filldraw[pattern=vertical lines]   (0,0) rectangle (6,1);
    \filldraw[pattern=vertical lines]   (0,5) rectangle (6,6);
    \filldraw[pattern=horizontal lines] (5,0) rectangle (6,6);
  \end{tikzpicture}
\]
\vspace*{-10mm}
  \caption{Geometric semantics of a simple program.}
  \label{fig:gsem}
\end{figure}
which consists of a square from which the grayed region has been removed (this
is called the \emph{forbidden region}). The horizontal axis corresponds to the
execution of the first process and the vertical one to the second process (we
have indicated the points corresponding to each instruction of the processes in
abscissa and ordinate in order to make this clear). A point in this space can
thus roughly be assimilated to a position in the program and the grayed removed
area corresponds to removing forbidden positions. A (continuous) path in this
space is said to be \emph{directed} when it is increasing with respect to each
component: such a path corresponds to an execution of the program. We sometimes
say a \emph{dipath} for a directed path. For instance, the dotted path on the left of \cref{fig:gsem} is
directed and corresponds to a total execution where the second process gets
wholly executed before the first one does. Finally, two dipaths are
\emph{dihomotopic} when the first can be continuously deformed to the second
while going through directed paths only. In general, the definition of the
geometric semantics is more involved than the above example shows (in
particular, the notion of direction is subtle in presence of loops) and is
detailed in the aforementioned references.

Let us now mention two main applications of geometric semantics. Firstly, under
simple \emph{coherence} assumptions, it can be shown that two dihomotopic
dipaths have the same effect on the state: in order to ensure the absence of
errors in a program (\eg we are never dividing by 0), it is thus sufficient to
use traditional techniques (\eg abstract
interpretation~\cite{cousot1977abstract}) for one representative in each
equivalence class. Secondly, it can be helpful to detect problems specific to
concurrency in programs. For instance, the black dot in the figure on the left
is a point which is not the final point (the upper right one) and is the source
of no non-trivial directed path: this indicates the presence of a
\emph{deadlock} in the program. Namely, if the first process executes $\P a$
then $\P b$ and the second process executes $\P b$ then $\P a$ the program is
locked and cannot proceed any further: the first process is waiting for the
second to release the mutex~$b$ and the second process is waiting for the first
to release the mutex~$a$.

We do not develop these techniques much further here -- they have already been
elsewhere -- and concentrate on the following question: how can we represent
this geometric semantics in practice? For programs consisting of $n$ processes in parallel, such as
the one above, where each process consists of a sequence of actions, the geometric
semantics will be an $n$-dimensional cube from which a finite number of
$n$-dimensional cubes have been carved out, forming the forbidden region (more
general cases are handled in~\cite{haucourt:geom-cons}). For instance, in our
example, the forbidden region consists of two cubes as shown in the figure in
the middle of \cref{fig:gsem}. This motivates investigating the notion of (\emph{cubical})
\emph{region}, which is a portion of the global cube obtained as a finite union
of cubes. Interestingly, the geometric semantics, which is the complement of the
forbidden region, can also be described as a union of cubes (8 in our example),
and is thus a region, as pictured on the right of \cref{fig:gsem}. In particular, by
convexity, any two dipaths with the same endpoints within a cube are
dihomotopic, which helps computing representatives in dihomotopy classes.

%
The state space associated to the above program is
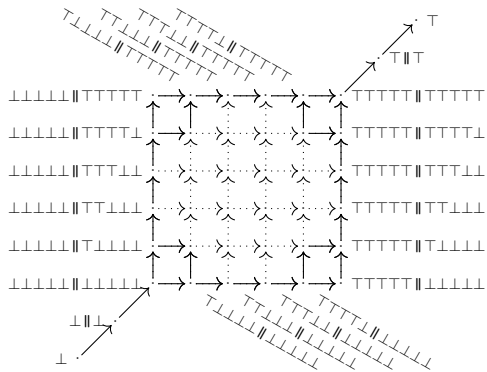
\[
  \begin{tikzpicture}[scale=.5, every node/.append style={font=\tiny},shorten >=2,shorten <=2]
    \coordinate (b) at (0,0);
    \draw (-2,-2) node {$\cdot$} node[left] {$\pbot$};
    \draw (-1,-1) node {$\cdot$} node[left] {$\ppar\pbot\pbot$};
    \draw[->] (-2,-2) -- (-1,-1);
    \draw[->] (-1,-1) -- (0,0);
    \foreach \x in {0,...,5} \foreach \y in {0,...,5} \draw (\x,\y) node{$\cdot$};
    \foreach \x in {0,...,5} \foreach \y in {0,...,4} \draw[->,dotted] (\x,\y) -- (\x,\y+1);
    \foreach \x in {0,...,4} \foreach \y in {0,...,5} \draw[->,dotted] (\x,\y) -- (\x+1,\y);
    \foreach \x in {0,5} \foreach \y in {0,...,4} \draw[->] (\x,\y) -- (\x,\y+1);
    \foreach \x in {0,...,4} \foreach \y in {0,5} \draw[->] (\x,\y) -- (\x+1,\y);
    \draw[->] (1,0) -- (1,1);
    \draw[->] (0,1) -- (1,1);
    \draw[->] (0,4) -- (1,4);
    \draw[->] (1,4) -- (1,5);
    \draw[->] (4,0) -- (4,1);
    \draw[->] (4,1) -- (5,1);
    \draw[->] (4,4) -- (5,4);
    \draw[->] (4,4) -- (4,5);
    \draw (0,0) node[left] {$\ppar{\pbot\pbot\pbot\pbot\pbot}{\pbot\pbot\pbot\pbot\pbot}$};
    \draw (1,0) node[below right] {\rotatebox{-30}{$\ppar{\ptop\pbot\pbot\pbot\pbot}{\pbot\pbot\pbot\pbot\pbot}$}};
    \draw (2,0) node[below right] {\rotatebox{-30}{$\ppar{\ptop\ptop\pbot\pbot\pbot}{\pbot\pbot\pbot\pbot\pbot}$}};
    \draw (3,0) node[below right] {\rotatebox{-30}{$\ppar{\ptop\ptop\ptop\pbot\pbot}{\pbot\pbot\pbot\pbot\pbot}$}};
    \draw (4,0) node[below right] {\rotatebox{-30}{$\ppar{\ptop\ptop\ptop\ptop\pbot}{\pbot\pbot\pbot\pbot\pbot}$}};
    \draw (0,1) node[left] {$\ppar{\pbot\pbot\pbot\pbot\pbot}{\ptop\pbot\pbot\pbot\pbot}$};
    \draw (0,2) node[left] {$\ppar{\pbot\pbot\pbot\pbot\pbot}{\ptop\ptop\pbot\pbot\pbot}$};
    \draw (0,3) node[left] {$\ppar{\pbot\pbot\pbot\pbot\pbot}{\ptop\ptop\ptop\pbot\pbot}$};
    \draw (0,4) node[left] {$\ppar{\pbot\pbot\pbot\pbot\pbot}{\ptop\ptop\ptop\ptop\pbot}$};
    \draw (0,5) node[left] {$\ppar{\pbot\pbot\pbot\pbot\pbot}{\ptop\ptop\ptop\ptop\ptop}$};
    \draw (6,6) node{$\cdot$} node[right]{$\ppar\ptop\ptop$};
    \draw (7,7) node{$\cdot$} node[right]{$\ptop$};
    \draw (5,0) node[right] {$\ppar{\ptop\ptop\ptop\ptop\ptop}{\pbot\pbot\pbot\pbot\pbot}$};
    \draw (5,1) node[right] {$\ppar{\ptop\ptop\ptop\ptop\ptop}{\ptop\pbot\pbot\pbot\pbot}$};
    \draw (5,2) node[right] {$\ppar{\ptop\ptop\ptop\ptop\ptop}{\ptop\ptop\pbot\pbot\pbot}$};
    \draw (5,3) node[right] {$\ppar{\ptop\ptop\ptop\ptop\ptop}{\ptop\ptop\ptop\pbot\pbot}$};
    \draw (5,4) node[right] {$\ppar{\ptop\ptop\ptop\ptop\ptop}{\ptop\ptop\ptop\ptop\pbot}$};
    \draw (5,5) node[right] {$\ppar{\ptop\ptop\ptop\ptop\ptop}{\ptop\ptop\ptop\ptop\ptop}$};
    \draw (1,5) node[above left] {\rotatebox{-30}{$\ppar{\ptop\pbot\pbot\pbot\pbot}{\ptop\ptop\ptop\ptop\ptop}$}};
    \draw (2,5) node[above left] {\rotatebox{-30}{$\ppar{\ptop\ptop\pbot\pbot\pbot}{\ptop\ptop\ptop\ptop\ptop}$}};
    \draw (3,5) node[above left] {\rotatebox{-30}{$\ppar{\ptop\ptop\ptop\pbot\pbot}{\ptop\ptop\ptop\ptop\ptop}$}};
    \draw (4,5) node[above left] {\rotatebox{-30}{$\ppar{\ptop\ptop\ptop\ptop\pbot}{\ptop\ptop\ptop\ptop\ptop}$}};
    \draw[->] (5,5) -- (6,6);
    \draw[->] (6,6) -- (7,7);
  \end{tikzpicture}
\]
(for simplicity, we have omitted some intermediate positions such as
$\ppar{(\pseq\pbot\pbot)}\pbot$, also we have omitted writing ``\code{;}'' in
the positions). The dotted edges are those which have to be removed in order to
obtain the pruned state space. One can see that, up to some minor details such
as the added positions at the beginning and at the end, the pruned graph can be thought of as an
``algebraic counterpart'' of the geometric semantics, which motivates the
investigation of performing the computations directly on this representation,
instead of going though an arbitrary encoding into spaces. In particular, we
define and study here an abstract axiomatization of regions, which applies to
such graphs.

\section{Regions of posets}
\label{sec:region}
In this section we define our representation of the state-space of programs,
replacing (maximal) cubical regions in geometric semantics~\cite{datc} by
(normal) \emph{syntactic regions}. We show that, under mild assumptions, these
regions satisfy the same fundamental properties as cubical regions (forming a
boolean algebra, supporting the existence of canonical representatives, etc.)
and provide explicit ways of computing corresponding canonical operations on
syntactic regions. To the best of our knowledge all the constructions performed
here are new, and based on the original notion of \emph{position} introduced
in~\cref{sec:prog}.

\subsection{Intervals}
Suppose fixed a poset $(X,\leq)$. An \emph{interval} $(x,y)$ in this poset is a
pair of elements such that~$x\leq y$. Given an interval~$I=(x,y)$, we write
$\support I=[x,y]=\setof{z\in X}{x\leq z\leq y}$ for the set of points it
contains, also called its \emph{support}. We write $z\in I$ instead of
$z\in\support I$: we have $z\in(x,y)$ iff $x\leq z\leq y$. We also write
$I\subseteq J$ instead of $\support I\subseteq\support J$: we have
$(x,y)\subseteq J$ iff $x\in J$ and~$y\in J$. Finally, we denote by
$\intervals(X)$ the poset of intervals of~$X$.

\subsection{Regions}
A \emph{region}~$R$ is a set of intervals: we write
$\regions(X)=\powset(X\times X)$ for the set of regions of~$X$, where
$\powset(-)$ denotes the powerset of a given set. A region is \emph{finite} when
it consists of a finite number of intervals. The \emph{support} of a region~$R$
is the set $\support R=\bigcup_{I\in R}\support I$ of points it contains.

Two regions are \emph{equivalent} when they have the same support. Such regions
are different ways of describing the same set of points: in order to be able to
correctly manipulate regions, we should be able to decide when two regions are
equivalent, and in order to be efficient, we should find the most compact
representation of a region in its equivalence class. Intuitively, the ``best''
region with a given support $Y\subseteq X$ consists of all the maximal intervals
(\wrt $\subseteq$) contained in $Y$. We will call this region the \emph{normal
  form} for a region~$R$ and say that a region is in normal form when it is
equal to its own normal form.
In order to formalize this, we introduce the following preorder on regions: we
write $\preceq$ for the pre-order on regions in~$\regions(X)$ defined by
\[
  R\preceq S
  \qquad
  \text{iff}
  \qquad
  \forall I\in R.\ \exists J\in S.\ I\subseteq J
\]
When $R\preceq S$ and $R$ and $S$ are equivalent, we think of $S$ as being a
``more economic way'' of describing the same support, because it uses bigger
intervals: every interval of~$R$ is contained in one of~$S$.


\begin{lemma}
  \label[lem]{lem:left-adj-preorder}
  The functions $\support-:\regions(X)\to\powset(X)$ and
  $\intervals:\powset(X)\to\regions(X)$ are increasing and $\support-$ is left
  adjoint to~$\intervals$.
\end{lemma}

\begin{example}
  \label[ex]{ex:region-order}
  Consider the set $X=[0,1]^2$ equipped with the usual partial order, and
  consider the three following regions on it, whose support is~$X$:
  \[
    \begin{array}{r@{\ }c@{\qquad}r@{\ }c@{\qquad}r@{\ }c}
      &
      \begin{tikzcd}
        \draw (0,0) rectangle (1,1);
        \filldraw[pattern=north east lines] (0,0) rectangle (.5,1);
        \filldraw[pattern=north west lines] (.5,0) rectangle (1,1);
      \end{tikzcd}
      &&
      \begin{tikzcd}
        \draw (0,0) rectangle (1,1);
        \filldraw[pattern=north east lines] (0,0) rectangle (1,1);
      \end{tikzcd}
      &&
      \begin{tikzcd}
        \draw (0,0) rectangle (1,1);
        \filldraw[pattern=north east lines] (0,0) rectangle (1,1);
        \filldraw[pattern=north west lines] (.25,.25) rectangle (.75,.75);
      \end{tikzcd}
      \\
      R_1=&\set{[0,\tfrac12]\times[0,1],[\tfrac12,1]\times[0,1]}
      &
      R_2=&\set{[0,1]\times[0,1]}
      &
      R_3=&R_2\cup\set{[\tfrac14,\tfrac34]\times[\tfrac14,\tfrac34]}
    \end{array}
  \]
  We have $R_1\prec R_2$ and $R_3\preceq R_2$, as well as $R_2\preceq R_3$.
\end{example}

\noindent
In the above example, the region $R_2$ is clearly the most parsimonious way to
represent the whole space, in the sense explained above, and is a maximal
element with respect to the order. However, there are many other maximal
elements such as $R_3$ (or, in fact, any other region obtained by adding
intervals to~$R_2$: the relation $\preceq$ is not antisymmetric). This motivates
the introduction of the following refined order on regions, which is such that
$R_1<R_3<R_2$:

\begin{definition}
  We define the relation $\leq$ on $\regions(X)$ by $R\leq S$ if and only if
  $R\preceq S$, and $S\preceq R$ implies $S\subseteq R$ (\ie every interval of
  $S$ belongs to $R$).
\end{definition}

\begin{lemma}
  \label[lem]{lem:po-region}
  The relation $\leq$ is a partial order.
\end{lemma}

\begin{lemma}
  \label[lem]{lem:sup-incr}
  The support function $\support-:\regions(X)\to\powset(X)$ is increasing if we
  equip the first with $\leq$ and second with $\subseteq$ as partial orders.
\end{lemma}

\noindent
We expect that the ``best description'' of a subset of~$X$ by a region is given
by a right adjoint $N:\powset(X)\to\regions(X)$ to the support function, which
to a subset~$Y$ of~$X$ associates the normal region describing it. Namely, the
adjoint functor theorem indicates that if the right adjoint exists it should satisfy
\begin{equation}
  \label{eq:N-sup}
  N(Y)=\bigvee\setof{R\in\regions(X)}{\support R\subseteq Y}
\end{equation}
However, such an adjoint is not always well-defined, as illustrated in
\cref{ex:no-N} below. The \emph{normal form} of a region~$R$ is, when defined,
the region $N(\support R)$, where $N$ is defined by the formula \eqref{eq:N-sup}
above. We say that a region is \emph{normalizable} when it admits a normal
form, and \emph{in normal form} when it is further equal to its normal form.

\begin{lemma}
  \label[lem]{lem:finite-norm-supp}
  Given $Y\subseteq X$ such that $N(Y)$ is defined, one has
  $\support{N(Y)}=Y$.
\end{lemma}

\begin{example}
  \label[ex]{ex:no-N}
  Consider the set~$X=[0,1]\subseteq\R$ equipped with the usual order. We claim
  that the subset $Y=X\setminus\set{1}$ does not have a normal form.
  By contradiction, suppose that the region~$N(Y)$ is well-defined. By the above \cref{lem:finite-norm-supp},
  $\support{N(Y)}=Y =[0,1[$. 
  Given~$ I \in N(Y)$, there exists $\varepsilon,\eta$ with $0\leq \varepsilon \leq 1-\eta < 1$ such that
  $I = [\varepsilon,1-\eta]$ and one easily checks that the region $R$ obtained
  from $N(Y)$ by removing this interval and replacing it by $[0,1-\eta/2]$ is
  such that $\support{R}=\support{N(Y)}=Y$ and $N(Y)<R$,
  contradicting~\eqref{eq:N-sup}.
  A very similar situation can be observed in the case of programs, by
  considering the program $P=\ploop A$ for some action~$A$. We write
  $X=\positions(P)$ for its poset of positions:
  \[
    \begin{tikzpicture}[shorten >=2,shorten <=2]
      \draw (0,0) node {$\cdot$} node[above] {$\pbot$};
      \draw (1,0) node {$\cdot$} node[above] {$\ploop[0]\pbot$};
      \draw (2,0) node {$\cdot$} node[above] {$\ploop[0]\ptop$};
      \draw (3,0) node {$\cdot$} node[above] {$\ploop[1]\pbot$};
      \draw (4,0) node {$\cdot$} node[above] {$\ploop[1]\ptop$};
      \draw (5,0) node {$\cdot$} node[above] {$\ploop[2]\pbot$};
      \foreach \x in {0,...,5} \draw[->] (\x,0) -- (\x+1,0);
      \draw[dotted] (6,0) -- (7,0);
      \draw (7,0) node {$\cdot$} node[above] {$\ptop$};
    \end{tikzpicture}
  \]
  The subset $Y=X\setminus\set{\ptop}$ does not have a normal form for similar
  reasons as above.
\end{example}

\begin{remark}
  In order to accommodate with situations such as in previous example, one could
  think of allowing (semi-)open intervals in addition to closed ones in
  regions. However, the operations on those quickly become very difficult to
  handle because it turns out that, in the case of programs of the form
  $\ppar PQ$, we need to be able to specify whether bounds of intervals are open
  or not for each component of the parallel composition.
\end{remark}

\section{The boolean algebra of finitely complemented regions}
\label{sec:boolean-algebra}
Given a set $Y\subseteq X$, we write $\bar{Y}=X\setminus Y$ for its \emph{complement}.
For practical applications, given a region~$R$ in~$\regions(X)$, we need to be
able to compute a region~$\ol R$ which covers the complement of the region, \ie
a region~$\ol R$ such that $\support{\ol R}=\ol{\support R}$. Typically, we have
explained in \cref{sec:mutex} that the pruned state space is obtained as the
complement in the state space of the forbidden region. Moreover, even when the
region~$R$ is not in normal form, we are able to compute the region~$\ol R$ in
normal form, which suggests that we should be able to compute the normal form of
a region~$R$ as $N(\support{R})=\ol{\ol R}$. Performing these computations will also
involve computing intersections and unions of regions, so that we will show that
-- under suitable hypothesis -- regions which admit a normal form have the
structure of a boolean algebra, providing explicit algorithmic constructions for
the corresponding operations.
This generalizes the situation considered in~\cite{datc} (which is limited to
regions which are subsets of~$\R^n$) and in~\cite{haucourt:geom-cons} (which is
limited to products of graphs). In order to ensure that our constructions are
applicable in practice, we only consider regions which are \ul{finite} in the
following (and showing that finiteness is preserved by our constructions will be
non-trivial). In particular, by a ``normal form'', we always mean a finite
region in normal form.

\subsection{The complement of an interval}
\label{sec:int-compl}
We begin by investigating the computation of the region corresponding to the
complement of an interval. As an illustration, consider the space
$X=[0,5]^2\subseteq\N^2$ and $R=\set{I}$ with $I=[(2,2),(3,3)]$, as pictured on
the left
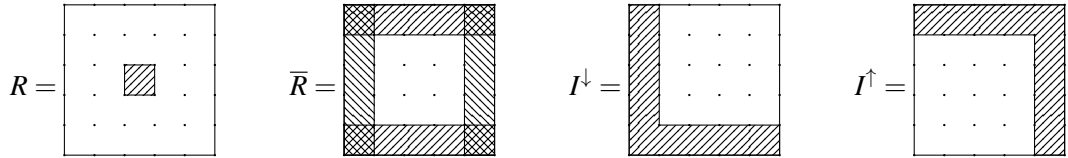
\begin{align*}
  R&=
  \begin{tikzpicture}[baseline=(b.base),scale=.4]
    \coordinate (b) at (2.5,2);
    \foreach \i in {0,...,5} \foreach \j in {0,...,5} \filldraw (\i,\j) circle (.02);
    \draw (0,0) rectangle (5,5);
    \draw[pattern=north east lines] (2,2) rectangle (3,3);
  \end{tikzpicture}
  &
  \ol R&=
  \begin{tikzpicture}[baseline=(b.base),scale=.4]
    \coordinate (b) at (2.5,2);
    \foreach \i in {0,...,5} \foreach \j in {0,...,5} \filldraw (\i,\j) circle (.02);
    \draw (0,0) rectangle (5,5);
    \draw[pattern=north east lines] (0,0) rectangle (5,1);
    \draw[pattern=north west lines] (0,0) rectangle (1,5);
    \draw[pattern=north east lines] (0,4) rectangle (5,5);
    \draw[pattern=north west lines] (4,0) rectangle (5,5);
  \end{tikzpicture}
  &
  \lcomp I&=
  \begin{tikzpicture}[baseline=(b.base),scale=.4]
    \coordinate (b) at (2.5,2);
    \foreach \i in {0,...,5} \foreach \j in {0,...,5} \filldraw (\i,\j) circle (.02);
    \draw (0,0) rectangle (5,5);
    \draw[pattern=north east lines] (0,0) -- (0,5) -- (1,5) -- (1,1) -- (5,1) -- (5,0) -- cycle;
  \end{tikzpicture}
  &
  \ucomp I&=
  \begin{tikzpicture}[baseline=(b.base),scale=.4]
    \coordinate (b) at (2.5,2);
    \foreach \i in {0,...,5} \foreach \j in {0,...,5} \filldraw (\i,\j) circle (.02);
    \draw (0,0) rectangle (5,5);
    \draw[pattern=north east lines] (0,5) -- (5,5) -- (5,0) -- (4,0) -- (4,4) -- (0,4) -- cycle;
  \end{tikzpicture}
\end{align*}
The normal form of its complement is the region
\[
  \ol R=\set{[\bot,y_1],[\bot,y_2],[x_1,\top],[x_2,\top]}
\]
with $\bot=(0,0)$, $\top=(5,5)$, $y_1=(1,5)$, $y_2=(5,1)$, $x_1=(0,4)$,
$x_2=(4,0)$. Here, it should be noted that $\ol R=\lcomp I\cup\ucomp I$ where
$\lcomp I$ (\resp $\ucomp I$) is the set of elements of~$X$ which are not above
$(2,2)$ (\resp below $(3,3)$), nor in fact above any element of~$I$. Moreover,
the points $y_1$ and $y_2$ are the maximal elements of the set~$\lcomp I$ and,
similarly, the points $x_1$ and $x_2$ are the minimal element of~$\ucomp I$. We
can generalize the situation as follows.

Given a set~$Y\subseteq X$, its \emph{lower closure} $\lc Y$ and \emph{lower
  complement}~$\lcomp Y$ are respectively the sets
\begin{align*}
  \lc Y&=\setof{x\in X}{\exists y\in Y.\ x\leq y}
  &
  \lcomp Y&=\setof{x\in X}{\forall y\in Y. x\not\geq y}
\end{align*}
The \emph{upper closure} $\uc Y$ and \emph{lower complement}~$\lcomp Y$ of~$Y$
are defined dually. We write $\max Y$
for the set of maximal elements of~$Y$, and $\min Y$ for the set of minimal
elements. The set~$Y$ is \emph{finitely lower generated} when there exists a
finite set $Y'$ such that $Y=\lc Y'$ (the notion of \emph{finitely upper
  generated} is defined dually). By extension, we say that an element $x$ of~$X$
is finitely lower (\resp upper) generated when $\set{x}$ is.

\begin{lemma}
  \label[lem]{lem:max-finite}
  If $Y$ is finitely lower generated then $\max Y$ is finite and we have $Y=\lc\max Y$.
\end{lemma}

\noindent
We say that a set $Y\subseteq X$ is \emph{finitely lower} (\resp \emph{upper})
\emph{complemented} when $\lcomp Y$ (\resp $\ucomp Y$) is finitely lower (\resp
upper) generated. By extension, we say that an element $x$ of~$X$
is finitely lower (\resp upper) complemented when $\set{x}$ is.
We say that $Y$ is \emph{finitely complemented} when it is
both finitely lower and upper complemented. We can finally characterize
intervals which admit a complement in normal form as follows:

\begin{proposition}
  \label[prop]{prop:int-normal-comp}
  Given a bounded lattice $X$ and $I$ an interval, the region $R=\set{I}$ has a
  complement in normal form if and only if $\support I$ is finitely
  complemented. In this case, the normal form of its complement is
  \[
    \ol R=\setof{[\bot,y]}{y\in\max(\lcomp I)}\cup\setof{[x,\top]}{x\in\min(\ucomp I)}
  \]
\end{proposition}
  %


\subsection{The complement of a region}
%
Our aim is now to generalize \cref{prop:int-normal-comp}, and characterize
normalizable regions (as opposed to intervals) which admit a complement in
normal form.
The condition which emerged in order to capture such situations is the following
one:

\begin{definition}
  \label[def]{def:fin-comp}
  A poset $(X,\leq)$ is \emph{finitely complemented} if
  \begin{enumerate}
  \item given a finitely lower complemented element $x\in X$, every element of
    $\max(\lcomp{\set{x}})$ is finitely upper complemented,
  \item given a finitely upper complemented element $x\in X$, every element of
    $\min(\ucomp{\set{x}})$ is finitely lower complemented.
  \end{enumerate}
\end{definition}

\begin{remark}
  \label[rem]{rem:wo-ug}
  In the case where $X$ is a well-order, every upwards closed subset is
  necessarily finitely upper generated: the set~$\min X$ of minimal elements
  of~$X$ generates~$X$ because it is well-founded, and is finite because it is
  an antichain. The first condition is thus always satisfied.
\end{remark}

We suppose fixed an ambient bounded lattice~$X$, in which the construction will
be performed. For technical reasons, it will be convenient to suppose that $X$
is also well-ordered, which, by \cref{prop:pos-wo}, is a reasonable restriction
for the applications we have in mind. We write
\[
  \nsupports(X)=\setof{Y\subseteq X}{\text{both $N(Y)$ and $N(\ol Y)$ exists and are finite}}
\]
for the poset (under inclusion) of \emph{normal supports} (\ie supports of
regions in normal forms, whose complements are also supports of regions in
normal forms).
Our goal is to show the following theorem:

\begin{theorem}
  \label{thm:nf-ba}
  The poset $\nsupports(X)$ is a boolean algebra if and only if $X$ is finitely
  complemented.
\end{theorem}

\noindent
One of the implications is easily shown:

\begin{proposition}
  \label[prop]{prop:nf-ba-easy}
  The left-to-right implication of \cref{thm:nf-ba} holds.
\end{proposition}
\begin{proof}
  By contraposition, suppose that $X$ is not finitely complemented. By
  \cref{rem:wo-ug}, the condition (i) in the \cref{def:fin-comp} is always
  satisfied, and therefore (ii) has to be falsified. This means there exists an
  upper complemented element $x\in X$ such that there exists an element
  $y\in\min(\ucomp{\set{x}})$ such that $y$ is not finitely lower generated. We show
  below that both $[\bot,y]$ and $\ol{[\bot,x]}$ belong to $\nsupports(X)$. But
  their intersection is not in $\nsupports(X)$: the set $\nsupports(X)$ is not
  closed under intersections and thus not a boolean algebra. Namely, an easy
  computation shows $[\bot,y]\cap\ol{[\bot,x]}=[y,y]$ and we have that
  $[y,y]\not\in\nsupports(X)$ by \cref{prop:int-normal-comp}, because $y$ is not
  finitely lower complemented.

  It remains to be shown that $[\bot,y]\in\nsupports(X)$. Since $y$ is upper
  complemented and $\bot$ is trivially lower complemented, we have that the
  interval $I=(\bot,y)$ is finitely complemented (by definition). Writing
  $R=\set{I}$, the region $R$ is trivially in normal form, and has a complement
  in normal form by \cref{prop:int-normal-comp}, and therefore $I$ belongs to
  $\nsupports(X)$. The proof that $\ol{[\bot,x]}\in\nsupports(X)$ is similar.
\end{proof}

We say that a region~$R$ is \emph{finitely complemented} if it contains only
intervals whose support is finitely complemented in the sense of
\cref{sec:int-compl}. Beware that this definition does not state that the
support of the region should be finitely complemented. We write
\[
  \fcsupports(X)=\setof{Y\subseteq X}{\text{$Y=\support R$ for some finitely complemented region $R$}}
\]
We also write
\[
  \fcregions(X)=\setof{R\in\regions(X)}{\text{$R$ is finitely complemented}}
\]
so that the elements of $\fcsupports(X)$ are the supports of regions in
$\fcregions(X)$.
The plan of our proof for the missing implication of \cref{thm:nf-ba} is as
follows: we first show that $\fcsupports(X)$ forms a boolean algebra by
explicitly constructing the required operations on finitely complemented regions
(\cref{prop:fcr-ba}) and then show that $\fcsupports(X)$ is isomorphic to
$\nsupports(X)$ (\cref{prop:fc-normal}).

In the rest of this section, we suppose that the poset $X$ is finitely
complemented.  Given two intervals $I=(x,y)$ and $I'=(x',y')$, we write
$(x,y)\cap(x',y')=(x\vee x',y\wedge y')$ for their intersection, which is always
their infimum (\wrt $\subseteq$ order) when defined (\ie
$x\vee x'\leq y\wedge y'$).

\begin{definition}
  \label[def]{def:boolean-operations}
  We define the following operations on regions $R,S$ in $\fcregions(X)$:
  \begin{itemize}
  \item union: $R\ncup S=R\cup S$
  \item intersection: $R\ncap S=\setof{I\cap J}{\text{$I\in R$, $J\in S$ and $I\cap J$ is defined}}$
  \item complement: $\ncomp{R}=\bigncap[(x,y)\in R](\setof{(\bot,y')}{y'\in\max(\lcomp{\set{x}})}\cup\setof{(x',\top)}{x'\in\min(\ucomp{\set{y}})})$
  \end{itemize}  
\end{definition}

\begin{lemma}
  \label[lem]{lem:def-op-bool}
  The above operations are well-defined on $\fcregions(X)$.
\end{lemma}
\begin{proof}
  The most subtle is intersection. It is shown by proving that finitely lower
  (\resp upper) complements of elements of $X$ are stable by $\vee$
  (\resp $\wedge$).
\end{proof}

\begin{lemma}
  \label[lem]{lem:op-commute-supp}
  The above operations are compatible with the corresponding ones on supports:
  for regions $R,S\in\fcregions(X)$, we have
  \begin{align*}
    [R\ncap S]&=[R]\cap[S]
    &
    [R\ncup S]&=[R]\cup[S]
    &
    [\ncomp{R}]&=\ol{[R]}
  \end{align*}
\end{lemma}

\noindent
The finitely complemented regions $\fcregions(X)$ thus form a sub-boolean
algebra of $\powset(X)$:

\begin{corollary}
  \label[prop]{prop:fcr-ba}
  The set $\fcsupports(X)$ is a boolean algebra.
\end{corollary}

\begin{proposition}
  \label[prop]{prop:fc-normal}
  Finitely complemented supports coincide with normal ones, \ie we have
  $\fcsupports(X)=\nsupports(X)$.
\end{proposition}
\begin{proof}
  $\fcsupports(X)\subseteq\nsupports(X)$. Given $Y\in\fcsupports(X)$, there
  exists a region~$R\in\fcregions(X)$ such that $\support{R}=Y$. By
  \cref{lem:def-op-bool}, its complement $\ncomp{R}$ belongs to~$\fcregions(X)$,
  since it can be computed using a finite number of unions and intersections of
  finitely complemented regions.
  Then, it is easy to prove $\max \ncomp{R} = N(\ol{\support{R}})=N(\ol
  Y)$. Applying this twice gives the normal form of $R$, and thus of~$Y$.

  $\fcsupports(X)\supseteq\nsupports(X)$.
  Suppose given $Y\in\nsupports(X)$. We are going to show
  $\ol Y\in\fcsupports(X)$ and by \cref{prop:fc-normal}, we will conclude
  $Y\in\fcsupports(X)$.
  We define an extension $\infcomp{-}:\regions(X)\to\regions(X)$ of $\ncomp{-}$
  for regions that are not finitely complemented by
  \[
    \infcomp{R}=\bigncap[(s,t) \in R]\setof{(\bot,x)}{x \in \lcomp s} \ncup  \setof{(x,\top)}{x \in \max \ucomp t}  
  \]
  and we show that $N(\ol Y)\subseteq\infcomp{N(Y)}$ and
  $\infcomp{N(Y)}\in\fcregions(X)$, by \cref{rem:wo-ug}. Thus
  $N(\ol Y)\in\fcregions(X)$, \ie $\ol Y\in\fcsupports(X)$ and finally
  $Y\in\fcsupports(X)$.
\end{proof}

\noindent
We have thus shown the other implication of \cref{thm:nf-ba}:

\begin{corollary}
  If $X$ is finitely complemented, the set $\nsupports(X)$ is a boolean algebra.
\end{corollary}

\noindent
This thus shows that, in a finitely complemented poset, we can implement the
usual boolean operations on regions while preserving the property of being
normalizable.

\section{Syntactic regions}
\label{sec:syntactic}
In the case of syntactic regions, \ie regions on the poset of positions of a
program, the operations of boolean algebra can be effectively implemented, by
induction on the structure of programs, following
\cref{def:boolean-operations}. Namely,
\begin{itemize}
\item the union of regions is immediate to implement,
\item the supremum and infimum of positions can be computed following
  \cref{prop:lat-op}, from which we can compute the intersection of regions,
\item we can compute the generators of the complement, from which we
  can compute the complement of regions.
\end{itemize}
Let us detail the last point. Given a position $p$ of a
program~$P$, we define, by induction on~$P$, the following set $\inf_P(p)$ of
positions of~$P$:
{\allowdisplaybreaks
  \begin{gather*}
    \begin{aligned}
      {\inf}_P(\pbot)&=\emptyset
      &
      {\inf}_{\pseq PQ}(\ptop)&=\set{\pseq\ptop\ptop}
      &
      {\inf}_{\por PQ}(\ptop)&=\set{\por\ptop\pbot,\por\pbot\ptop}
      \\
      {\inf}_{A}(\ptop)&=\set{\pbot}
      &
      {\inf}_{\ppar PQ}(\ptop)&=\set{\ppar\ptop\ptop}
      &
      {\inf}_{\ploop P}(\ptop)&=\emptyset
    \end{aligned}
    \\
    \begin{aligned}
      {\inf}_{\pseq PQ}(\pseq pq)&=
      \begin{cases}
        \set{\pbot}&\text{if $p=q=\pbot$}\\
        \setof{\pseq{p'}\pbot}{p'\in\inf_P(p)}&\text{if $p\neq\pbot$ and $q=\pbot$}\\
        \setof{\pseq\ptop{q'}}{q'\in\inf_Q(q)}&\text{if $p=\ptop$ and $q\neq\pbot$}
      \end{cases}
      \\
      {\inf}_{\por PQ}(\por pq)&=
      \begin{cases}
        \set{\pbot}&\text{if $p=q=\pbot$}\\
        \setof{\por{p'}q}{p'\in\inf_P(p)}\cup\set{\por p\ptop}&\text{if $p\neq\pbot$}\\
        \setof{\por p{q'}}{q'\in\inf_Q(q)}\cup\set{\por\ptop q}&\text{of $q\neq\pbot$}
      \end{cases}
      \\
      {\inf}_{\ppar PQ}(\ppar pq)&=
      \begin{cases}
        \set{\pbot}&\text{if $p=q=\pbot$}\\
        \setof{\ppar{p'}\ptop}{p'\in\inf_P(p)}\cup\setof{\ppar\ptop{q'}}{q'\in\inf_Q(q)}&\text{if $p\neq\pbot$ or $q\neq\pbot$}
      \end{cases}
      \\
      {\inf}_{\ploop P}(\ploop[n]p)&=
      \begin{cases}
        \set{\pbot}&\text{if $n=0$ and $p=\pbot$}\\
        \set{\ploop[n-1]{\ptop}}&\text{if $n>0$ and $p=\pbot$}\\
        \setof{\ploop[n]{p'}}{p'\in\inf_P(p)}&\text{if $p\neq\pbot$}
      \end{cases}
    \end{aligned}
  \end{gather*}
}

\begin{proposition}
  \label[prop]{prop:max-lc-algo}
  Given a position~$p$ of a program~$P$, the set $\inf_P(p)$ is a well-defined
  set of positions and we have $\inf_P(p)=\max(\lcomp p)$.
\end{proposition}


\noindent
Similarly, given a position $p$ of a program~$P$, one can define by induction
on~$P$ a set $\sup_P(p)$ of positions of~$P$ such that
$\sup_P(p)=\min(\ucomp p)$. We can finally show that the poset of positions of a
programs satisfy the conditions of previous section:

\begin{proposition}
  \label[prop]{prop:pos-fc-wo}
  Given a program~$P$, its poset of positions $\positions(P)$ is a finitely
  complemented well-ordered lattice.
\end{proposition}

The operations do not in general preserve the property of being normal for
regions, but \cref{thm:nf-ba,prop:pos-fc-wo} however ensure that the property of
being normalizable is preserved, and the normal form of a region can be computed
as follows:

\begin{proposition}
  With the implementation of operations described above, the normal form of a
  region~$R$ is $N(R)=\max(\ol{\ol R})$, \ie it can be obtained by computing
  twice the complement of~$R$ and only keeping intervals which are maximal \wrt
  inclusion.
\end{proposition}
\begin{proof}
  It can be observed that the definition of the complement is such that it
  contains all the maximal intervals.
\end{proof}

\begin{example}
  Let us illustrate the fact that the operations defined above do not preserve
  normality (again, they only preserve normalizability), consider the following
  examples. With the region~$R$ and~$S$ below, the region $R\ncup S$ is not
  normal (the normal form is show on the right):
  \[
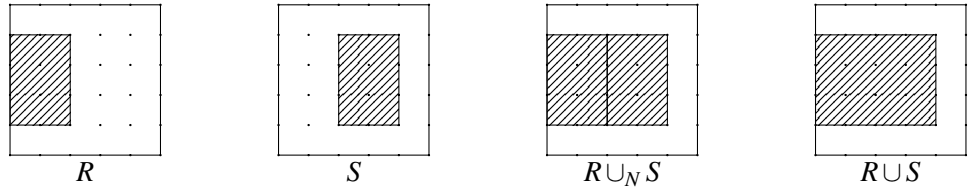
    \begin{array}{c@{\qquad\qquad}c@{\qquad\qquad}c@{\qquad\qquad}c}
      \begin{tikzpicture}[baseline=(b.base),scale=.4]
        \coordinate (b) at (2.5,2);
        \foreach \i in {0,...,5} \foreach \j in {0,...,5} \filldraw (\i,\j) circle (.02);
        \draw (0,0) rectangle (5,5);
        \draw[pattern=north east lines] (0,1) rectangle (2,4);
      \end{tikzpicture}
      &
      \begin{tikzpicture}[baseline=(b.base),scale=.4]
        \coordinate (b) at (2.5,2);
        \foreach \i in {0,...,5} \foreach \j in {0,...,5} \filldraw (\i,\j) circle (.02);
        \draw (0,0) rectangle (5,5);
        \draw[pattern=north east lines] (2,1) rectangle (4,4);
      \end{tikzpicture}
      &
      \begin{tikzpicture}[baseline=(b.base),scale=.4]
        \coordinate (b) at (2.5,2);
        \foreach \i in {0,...,5} \foreach \j in {0,...,5} \filldraw (\i,\j) circle (.02);
        \draw (0,0) rectangle (5,5);
        \draw[pattern=north east lines] (0,1) rectangle (2,4);
        \draw[pattern=north east lines] (2,1) rectangle (4,4);
      \end{tikzpicture}
      &
      \begin{tikzpicture}[baseline=(b.base),scale=.4]
        \coordinate (b) at (2.5,2);
        \foreach \i in {0,...,5} \foreach \j in {0,...,5} \filldraw (\i,\j) circle (.02);
        \draw (0,0) rectangle (5,5);
        \draw[pattern=north east lines] (0,1) rectangle (4,4);
      \end{tikzpicture}
      \\
      R
      &
      S
      &
      R\ncup S
      &
      R\cup S
    \end{array}
  \]
  Similarly, with the region $R$ and $S$ below, the region $R\ncap S$ is not
  normal because it contains the interval~$I$ pictured on the right
  \[
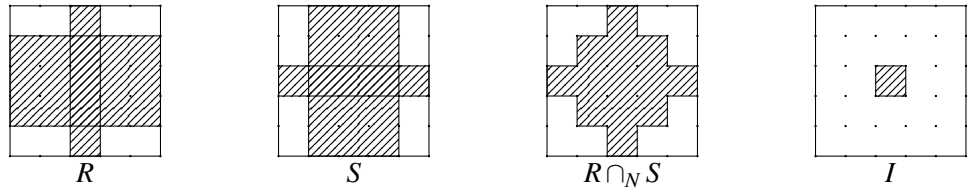
    \begin{array}{c@{\qquad\qquad}c@{\qquad\qquad}c@{\qquad\qquad}c}
    \begin{tikzpicture}[baseline=(b.base),scale=.4]
      \coordinate (b) at (2.5,2);
      \foreach \i in {0,...,5} \foreach \j in {0,...,5} \filldraw (\i,\j) circle (.02);
      \draw (0,0) rectangle (5,5);
      \draw[pattern=north east lines] (2,0) rectangle (3,5);
      \draw[pattern=north east lines] (0,1) rectangle (5,4);
    \end{tikzpicture}
    &
    \begin{tikzpicture}[baseline=(b.base),scale=.4]
      \coordinate (b) at (2.5,2);
      \foreach \i in {0,...,5} \foreach \j in {0,...,5} \filldraw (\i,\j) circle (.02);
      \draw (0,0) rectangle (5,5);
      \draw[pattern=north east lines] (0,2) rectangle (5,3);
      \draw[pattern=north east lines] (1,0) rectangle (4,5);
    \end{tikzpicture}
    &
    \begin{tikzpicture}[baseline=(b.base),scale=.4]
      \coordinate (b) at (2.5,2);
      \foreach \i in {0,...,5} \foreach \j in {0,...,5} \filldraw (\i,\j) circle (.02);
      \draw (0,0) rectangle (5,5);
      \filldraw[pattern=north east lines] 
      (2,0) -- (3,0) -- (3,1) -- (4,1) -- (4,2) -- (5,2) --
      (5,3) -- (4,3) -- (4,4) -- (3,4) -- (3,5) -- (2,5) --
      (2,4) -- (1,4) -- (1,3) -- (0,3) -- (0,2) -- (1,2) --
      (1,1) -- (2,1) -- (2,0) ;
    \end{tikzpicture}
    &
    \begin{tikzpicture}[baseline=(b.base),scale=.4]
      \coordinate (b) at (2.5,2);
      \foreach \i in {0,...,5} \foreach \j in {0,...,5} \filldraw (\i,\j) circle (.02);
      \draw (0,0) rectangle (5,5);
      \draw[pattern=north east lines] (2,2) rectangle (3,3);
    \end{tikzpicture}
    \\
    R
    &
    S
    &
    R\ncap S
    &
    I
    \end{array}
  \]
  (the normal form contains $3$ intervals which do not include~$I$).
\end{example}

\section{Future work}
A toy implementation of the computations described in this paper can be tested
online at~\cite{git}. It allows computing the forbidden region, the state space
(called there the \emph{fundamental region}) and the deadlocks of a program with
loops (we also plan to implement of further analysis of programs, handling
values with abstract domains as explained in the introduction). The positions
for loops are defined there \emph{without} unfolding: this allows handling the
case of forbidden regions within loops, but the theory is more involved (the
execution relation on position does not induce a partial order anymore) and left
for future work. Various practical examples of concurrent programs are given on
the website and the reader is welcome to try out some more of his own. We plan
to investigate, in a near future, the investigation of the combination of this
approach with abstract interpretation techniques in order to be able to
meaningfully handle domains of values.
Finally, we also plan to perform a formalization (in Agda) of the theory
developed there in order to make sure that no corner case is omitted (as it can
be observed in \cref{sec:syntactic}, the operations are defined by case
analysis, requiring to distinguish many possibilities and the situation is even
worse in generalizations).

\bibliography{papers}
\end{document}